\documentclass[twoside]{article}
\usepackage[english]{babel}
\usepackage[utf8]{inputenc}
\usepackage[T2A]{fontenc}
\usepackage{amssymb,cite}
\usepackage{amsmath}
\usepackage{amsthm}
\usepackage{graphicx}
\usepackage{hhline}
\usepackage{euscript}
\usepackage{geometry} 
\usepackage{indentfirst}
\usepackage{amsbib}

\usepackage{float}

\theoremstyle{plain}
\newtheorem{theorem}{Theorem}
\newtheorem{lemma}{Lemma}
\newtheorem{corollary}{Corollary}

\theoremstyle{definition}
\newtheorem{definition}{Definition}
\newtheorem{remark}{Remark}
\newtheorem{example}{Example}
\newtheorem{constraint}{Constraint}

\newcommand\restr[2]{{% we make the whole thing an ordinary symbol
  \left.\kern-\nulldelimiterspace % automatically resize the bar with \right
  #1 % the function
  \vphantom{\big|} % pretend it's a little taller at normal size
  \right|_{#2} % this is the delimiter
  }}

\AtBeginDocument{\def\refname{References}}
\AtBeginDocument{}

\begin{document}

\title{The complexity of multilayer $d$-dimensional circuits}

\author{
    \\ Sitdikov Timur Rashidovich\footnote{Google LLC, London, UK}
	\\ Kalachev Gleb Vyacheslavovich\footnote{Lomonosov Moscow State University, Moscow, Russia}
}
\date{}

\maketitle

\begin{abstract}
In this paper we research a model of multilayer circuits with a single logical layer. We consider $\lambda$-separable graphs as a support for circuits. We establish the Shannon function lower bound $\max \left(\frac{2^n}{n}, \frac{2^n (1 - \lambda)}{\log k} \right)$ for this type of circuits where $k$ is the number of layers. For $d$-dimensional graphs, which are $\lambda$-separable for $\lambda = \frac{d - 1}{d}$, this gives the Shannon function lower bound $\frac{2^n}{\min(n, d \log k)}$. For multidimensional rectangular circuits the proved lower bound asymptotically matches to the upper bound.
\end{abstract}

\newpage

\begin{center}
\tableofcontents
\vspace{3cm}
\end{center}

\clearpage

\newpage

\section{Introduction}

The problem of designing circuits which compute Boolean functions and are optimal or suboptimal in some sense appeared in the middle of the 20th century due to the rapid development of computer technology. One of the most intensive studied circuit models since the 1950s is Boolean circuits. The number of gates (also \emph{size} or \emph{complexity}) is a natural complexity measure for Boolean circuits. One may define the complexity of a Boolean function as the minimal size of a Boolean circuit computing the function. Muller \cite{Muller1} showed that the size of every Boolean function of $n$ variables does not exceed $O \left(\frac{2^n}{n}\right)$. Lupanov \cite{Lupanov1} proved that the complexity of almost all Boolean functions over the standard basis $\{\vee,\&,\neg\}$ is asymptotically equal to $\frac{2^n}{n}$. Also Lupanov obtained asymptotic bound for the complexity of Boolean functions with respect to any finite basis.

In practice when designing Boolean circuits one must take into account several factors like placement of gates, wiring, etc. Models of Boolean circuits considering these factors to some extent were studied in some papers appeared since the 1960s. Korshunov \cite{Korshunov1} obtained size bounds for Boolean circuits embedded in a 3-dimensional space with lower-bounded distances between gates, lower-bounded distances between wires and upper-bounded lengths of wires. Kravtsov \cite{Kravtsov1} considered Boolean circuits with gates placed in cells of a rectangular grid and proved the order of $2^n$ for the Shannon function. McColl \cite{McColl1} obtained Shannon function lower bound $\Omega(2^n)$ for planar circuits.

Models of cellular circuits similar to Kravtsov's model were considered in several more recent papers. Albrecht \cite{Albrecht1} showed that Shannon function asymptotics for cellular circuits has a form $c \cdot 2^n$, where $c$ is a constant dependent on a basis. Gribok \cite{Gribok1} obtained Shannon function asymptotics to $2^n$ for a special basis of cellular elements. The connection between size and other complexity measures for cellular circuits also has been examined. Cheremisin \cite{Cheremisin1} showed that it's impossible to design a cellular circuit of optimal size and activity simultaneously for a binary decoder. Kalachev \cite{Kalachyov0, Kalachyov1, Kalachyov2, Kalachyov3, Kalachyov4} researched simultaneous minimization of a size, depth and activity for cellular circuits. Efimov \cite{Efimov1, Efimov2} examined potential of three-dimensional cellular circuits.

VLSI circuits are one of the closest to practice circuit models. In VLSI circuits length of wires define the signal propagation time between gates. VLSI circuits have been considered in a number of papers and books (Thompson \cite{Thompson1}, Ullman \cite{Ullman1}). Kramer and van Leeuwen \cite{Kramer1} researched simultaneous minimization of size (area) and period.

Another direction of research is a connection between complexity measures for different circuit models. Savage \cite{Savage1, Savage2} examined the connection between VLSI circuits and planar circuits. Shkalikova \cite{Shkalikova1} showed a relation between the area of flat circuits and the volume of three-dimensional circuits.

The bounds of size proved for the mentioned above circuit models are above Lupanov's bound $\frac{2^n}{n}$ for Boolean circuits size. One of the reasons for this difference is that it's impossible to conduct arbitrary number of wires between gates under spatial constraints. If Boolean circuits are embedded into a graph (e.g. rectangular grid), the number of wires that can be conducted between fragments of the graph is naturally bounded by the size of edge separator in the graph.

In this paper we examine the relation between Shannon function and separability properties of the graph where Boolean circuits are embedded. We consider embeddings with constraints from \cite{KModel1}:
\begin{itemize}
	\item{No more than one non-trivial gate of a Boolean circuit can be embedded into any vertex of a graph.}
	\item{No more than $k$ wires of a Boolean circuit can be embedded into any edge of a graph.}
\end{itemize}

The main result of this paper is the lower bound of Shannon function for graphs with separator of size $O(p^{\lambda})$, where $p$ is the order of a graph and $0 < \lambda < 1$. We call such graphs as $\lambda$-separable. We also show that the proved lower bound is applicable to circuits embedded into a space with $2$ or more dimensions. Given the Shannon function upper bound for multidimensional rectangular circuits \cite{KModel1}, we obtain the Shannon function asymptotics for this model of circuits.

\section{Key definitions and results}

\subsection{Multilayer circuits}

The model if multilayer circuits with a single logical layer was considered in \cite{KModel1}. Let's briefly summarize key definitions. 

According to \cite[p.~148]{Chashkin1}, a Boolean circuit in a basis $B$ is a labeled directed acyclic graph. The labeling of vertices defines which vertices are inputs or outputs. It also maps all non-input vertices to Boolean functions from the basis $B$. Edges (wires) of a Boolean circuit are labeled by integers, and for each vertex the labeling of its input edges defines the order of arguments for the Boolean function mapped to the vertex.

A \emph{support} is a nonempty graph with a finite or countable number of vertices. In general a support may contain multiple edges or self-loops.

An \emph{embedding} of a Boolean circuit $S$ into a support $T$ is a homomorphism $h \colon S \to T$.

A \emph{circuit with a support} $T$ is a pair $(S,h)$ where $S$ is a Boolean circuit and $h$ is an embedding of $S$ into $T$. We use the term "circuit" instead of "circuit with a support" for brevity where it would not cause a misreading. A circuit $(S, h)$ computes a Boolean function $f$ if a Boolean circuit $S$ computes $f$.

In practical terms these definitions may be interpreted as follows. One of the problems in VLSI design is an embedding of gates and wires into a plate. The plate may be considered as a graph, i.e. as a support for Boolean circuits.

Usually there are some constraints on embeddings in circuit design problems. In this paper we consider the following constraints.
\begin{constraint}\label{CondVertex}
	Any vertex of a support may contain no more than one gate computing non-constant Boolean function.
\end{constraint}
\begin{constraint}\label{CondEdge}
    Any edge of a support may contain no more than $k$ wires of a Boolean circuit.
\end{constraint}

These constraints may be interpreted as follows. A circuit consists of $k$ "layers" where only one layer is "logical" (i.e. may contain gates computing non-constant Boolean functions). The remaining layers are used only for wiring. Therefore we call the circuits under constraints \ref{CondVertex}--\ref{CondEdge} as \emph{multilayer circuits}.

Let's denote by $M_k^T$ the set of all $k$-layer circuits with a support $T$.

The \emph{complexity} of a multilayer circuit if the number of support vertices used in the corresponding embedding. If $M$ is a set of multilayer circuits and $f$ is a Boolean function, one may naturally define the complexity of the function $f$ in the set $M$ as the minimal complexity of a circuit from $M$ computing $f$. If no such a circuit exists in $M$, we may formally consider infinity as the complexity of $f$. Let's denote the complexity of the function $f$ in the set $M$ as $L(M, f)$.

One may naturally define the Shannon function of the complexity of $k$-layer circuits with the support $T$:
\[
	L(M_k^T, n) = \max_{f \in B_n} L(M_k^T,f).
\]

\subsection{Supports}
The properties of a support are crucial for embeddings, as under the same constraints different supports in general admit completely different sets of embeddings. In this paper we consider $\lambda$-separable graphs as supports. We also consider $d$-dimensional graphs as an important special case of $\lambda$-separable graphs.

A class $\mathcal{G}$ of graphs is \emph{monotone} if every subgraph of a graph in $\mathcal{G}$ is also in $\mathcal{G}$.

\subsubsection{Classes of graphs $\mathcal{G}(q, \theta)$ and $\mathcal{G}(\lambda, q, \theta)$}
Let $q \in \mathbb{N}$ and $\theta > 1$ be some constants. Let's define class of graphs $\mathcal{G}(q, \theta)$ as the set of all supports with the following properties:
\begin{itemize}
	\item{Degree of each vertex in $T$ is bounded by $q$.}
	\item{For any integer $p$ the number of different non-isomorphic subgraphs of $T$ with $p$ vertices does not exceed $\theta^p$.}
\end{itemize}

The first property (bounding for vertex degree) is a natural limitation for circuit design problems. The second property is met for several important categories of graphs, including planar graphs \cite{PlanarLowerBound1} and $d$-dimensional graphs defined below.

The formal definition of $\lambda$-separability is considered in the section \ref{SecSep}. Substantively each subgraph of a $\lambda$-separable support can be split into smaller fragments by removing $O(p^{\lambda})$ vertices (edges), where $p$ is the number of vertices in the subgraph and $0 < \lambda < 1$.

Let's denote the subclass of $\lambda$-separable supports from $\mathcal{G}(q, \theta)$ as $\mathcal{G}(\lambda, q, \theta)$.

\subsubsection{$d$-dimensional graphs}
Let $d \ge 2$ be an integer. A support $T$ is a \emph{$d$-dimensional graph}, if there are constants $c_v > 0, c_e > 0$ such that $T$ can be embedded into $d$-dimensional Euclidean space with pairwise distances between vertices no less than $c_v$ and edge lengths no greater than $c_e$.

\begin{remark}
The constraints in the definition above are similar to the constraints in the definition of circuits with volumetric gates from Korshunov's paper \cite{Korshunov1}.
\end{remark}

\begin{remark}
We can always assume that one of the constants $c_v$ and $c_e$ is equal to $1$. Below we assume that $c_v = 1$.
\end{remark}

\begin{remark}
It's clear that every finite support is a $d$-dimensional graph with a great enough value of $c_e$. Therefore the definition of $d$-dimensional graphs is senseless for finite supports. However one may define a \emph{monotone class of $d$-dimensional graphs with a parameter $c_e$}, where the constant $c_e$ is common for all graphs in the class. It's obvious that such a class is a monotone class of graphs. Below we omit the constant $c_e$ and speak about a \emph{$d$-dimensional class of graphs} in cases when the value of $c_e$ is not important.
\end{remark}

\begin{example}
The graph of a $d$-dimensional grid is a $d$-dimensional graph. It's sufficient to consider $c_e = 1$.
\end{example}

\begin{example}
One can prove that the graph of an infinite binary tree is not $d$-dimensional for any $d$. Indeed, the number of vertices at distance $p$ from the root depends on $p$ exponentially, though the number of $d$-dimensional balls with radii $1$ that can be placed into a ball with radii $c_e \cdot d$ is $O(p^d)$.
\end{example}

Embedding of Boolean circuits into $d$-dimensional grid was considered in \cite{KModel1}. As in that paper, we use the term \emph{multidimensional rectangular circuits} for such circuits and use the notation $M_k^d$ instead of $M_k^{\mathbb{Z}^d}$.

In section \ref{SecDDimensional} we prove that all $d$-dimensional supports belong to classes $\mathcal{G}(\lambda, q, \theta)$ for $\lambda = \frac{d - 1}{d}$ and some values of $q$ and $\theta$.

\subsection{Other designations and agreements}

The expression $\log a$ always denotes a base two logarithm of $a$. We formally assume that $x \log x = 0$ for $x = 0$.

We denote by $B_{n,m}$ the set of Boolean functions with $n$ inputs and $m$ outputs ($n \ge 0$, $m \ge 1$).

The expression $f(x) \lesssim g(x)$ corresponds to the inequality $\varlimsup_{x \to \infty} \frac{f(x)}{g(x)} \le 1$. Similarly we use the expression $f(x) \gtrsim g(x)$. We may use a complex condition when passing to a limit, e.\,g. $f(n, k) \lesssim g(n, k)$ as $k \to \infty$, $\log k \le n$.

\subsection{Results}

In this paper for every support $T \in \mathcal{G}(\lambda, q, \theta)$ we prove that
\[
	L(M_k^T, n) \gtrsim \max \left(\frac{2^n}{n}, \frac{2^n (1 - \lambda)}{\log k} \right)\quad\mbox{as}\quad k\to\infty,\ n\to\infty.
\]

It is also proved that every $d$-dimensional support belongs to the class $\mathcal{G}(\lambda, q, \theta)$ for $\lambda = \frac{d - 1}{d}$ and some constants $q$ and $\theta$. Therefore the following estimation holds for for $d$-dimensional supports:
\[
	L(M_k^T, n) \gtrsim \frac{2^n}{\min(n,d \log k)}\quad\mbox{as}\quad k\to\infty,\ n\to\infty.
\]

An upper bound of Shannon function for multidimensional rectangular circuits matching the lower bound above was proved in \cite{KModel1}. Thus we have the asymptotics of Shannon function for multidimensional rectangular circuits:
\[
	L(M_k^d, n) \sim \frac{2^n}{\min(n,d \log k)}\quad\mbox{as}\quad k\to\infty,\ n\to\infty.
\]

\subsection{The structure of the paper}

In this paper all the proofs are divided into three sections.

Section \ref{SecSep} contains definitions related to graph separators. The main result of the section is lemma \ref{LemSepDiv}. The point of the lemma is that $\lambda$-separable graphs supporting "good" (in some sense) partitioning into two parts also support "good" partitioning into many parts.

Section \ref{SecPolSepGraph} contains the proof of the lower bound for Shannon function of the complexity for supports from classes $\mathcal{G}(\lambda, q, \theta)$. The key part of this section is lemma \ref{LemLowerBoundCount}.

Section \ref{SecDDimensional} is devoted to the proof of the lower bound for Shannon function for $d$-dimensional supports. The section also contains the asymptotics of Shannon function for multidimensional rectangular circuits as a corollary. Essentially it's proved that every $d$-dimensional support belongs to a class $\mathcal{G}(\lambda, q, \theta)$ for $\lambda = \frac{d - 1}{d}$ and some constants $q$ and $\theta$.

\section{Graph separators and their properties}\label{SecSep}

\subsection{Definitions and the simplest properties of separators}

In this section we provide the definitions of edge and vertex separators in graphs and prove some of the simplest properties of separators.

\paragraph{Edge separators.}
We define edge separators similarly to the definitions of vertex separators from \cite{PlanarSeparatorTheorem}.

\begin{definition}\label{DefEdgeSep}
    Let $f \colon \mathbb{N} \to \mathbb{R}$ be a function. A monotone class of graphs $\mathcal{G}$ is \emph{edge $f(p)$-separable} if there exist constants $\frac{1}{2} \le \alpha < 1$, $\beta \ge 0$, $m \ge 2$ such that any graph $G \in \mathcal{G}$ with $p$ vertices ($p \ge m$) can be split into two subgraphs with no more than $\alpha p$ vertices each and no more than $\beta f(p)$ edges between the subgraphs.
\end{definition}

\begin{remark}
The constant $m$ is technically important, since it allows not to consider some corner cases. For example graph $K_1$ cannot be split into two nonempty subgraphs in principle, thus we may always assume that $m \ge 2$. In general $m$ may be greater than $2$.
\end{remark}

\begin{definition}
    Let $f \colon \mathbb{N} \to \mathbb{R}$ be a function. A support $T$ is \emph{edge $f(p)$-separable} if the monotone class of all finite subgraphs of $T$ is edge $f(p)$-separable.
\end{definition}

The interesting case is when $f(p)$ is a slowly growing function. Essentially this allows to use the divide-and-conquer technique for obtaining effective algorithms and non-trivial lower bounds in proofs. In this paper we consider the function $p^{\lambda}$ with $0 < \lambda < 1$ as $f(p)$. We also call edge $p^{\lambda}$-separable supports and monotone classes of graphs as \emph{edge $\lambda$-separable}.

\paragraph{Vertex separators.}
The following definition of a vertex separator is a modification of definition 2.1 from \cite{MultidimensionalSeparator} applied to monotone classes of graphs.

\begin{definition}\label{DefVertexSep}
    Let $f \colon \mathbb{N} \to \mathbb{R}$ be a function. A monotone class of graphs $\mathcal{G}$ is \emph{vertex $f(p)$-separable} if there exist constants $\frac{1}{2} \le \alpha < 1$, $\beta \ge 0$, $m \ge 2$ such that for any graph $G \in \mathcal{G}$ with $p$ vertices ($p \ge m$) there exists a partition of $V(G)$ into three parts $A$, $B$, $C$ satisfying the following conditions:
\begin{itemize}
	\item{There are no edges from $A$ to $B$.}
	\item{$|A|, |B| \le \alpha p$.}
	\item{$|C| \le \beta f(p)$.}
\end{itemize}
\end{definition}

It's obvious that for any monotone class of graphs edge $f(p)$-separability implies vertex $f(p)$-separability, as one may consider endpoints of an edge separator as a vertex separator. The converse is not always true. For example, the class of stars $K_{1, p}$ and their subgraphs is vertex $1$-separable, but is not edge $1$-separable.

The following simple lemma shows that for monotone classes of graphs with bounded vertex degree vertex $f(p)$-separability implies edge $f(p)$-separability.

\begin{lemma}\label{AssSepVtoE}
    Let $\mathcal{G}$ be a monotone class of vertex $f(p)$-separable graphs with parameters $\alpha$, $\beta$, $m$ where vertex degree of any graph is bounded by $q$. Then $\mathcal{G}$ is edge $f(p)$-separable with parameters $\max \left(\frac{2}{3}, \alpha \right)$, $q \beta$ and $\max(m, 2)$.
\end{lemma} 
\begin{proof}
We'll show how to obtain an edge separator from a vertex separator.

Let $G \in \mathcal{G}$, $|V(G)| = p \ge \max(m, 2)$. By the definition of vertex separability $V(G)$ can be divided into three sets $A, B, C$, where $C$ is a separator. Here $|A|, |B| \le \alpha p$, $|C| \le \beta f(p)$.

Let's move vertices from $C$ to $A$ and $B$ in a way to keep the sizes of the resulting sets as close to each other as possible. Let's denote the resulting sets by $A'$ and $B'$. Considering the way of constructing $A'$ and $B'$, we obtain $1 \le |A'|, |B'| \le \max \left(\frac{2}{3}, \alpha \right) \cdot p$.

Each edge connecting $A'$ and $B'$ is incident to at least one vertex from $C$. Since vertex degree is bounded by $q$, the total number of such edges does not exceed $q|C| \le q \beta f(p)$.

Since values $\max \left(\frac{2}{3}, \alpha \right)$, $q \beta$ and $\max(m, 2)$ do not depend on a graph, edge $f(p)$-separability of $\mathcal{G}$ is proved.
\end{proof}

Thereby when we define a class $\mathcal{G}(\lambda, q, \theta)$ it does not matter whether we use edge $\lambda$-separability or vertex $\lambda$-separability, as all graphs from the class have vertex degree bounded by $q$.

\subsection{Partitioning of $\lambda$-separable graphs}

Informally the key result of this section is the following statement. Since a $\lambda$-separable graph can be split into two disconnected parts of comparable size by removing a small number of edges, the graph can also be split into many disconnected parts of bounded size by removing a small number of edges.

The following lemma is the modification of lemma 1 from \cite{RDivision1} for planar graphs.

\begin{lemma}\label{LemSepDiv}
    Let $\mathcal{G}$ be a monotone class of edge $\lambda$-separable graphs with parameters $\alpha$, $\beta$ and $m$, where $0 < \lambda < 1$, $\frac{1}{2} \le \alpha < 1$, $\beta \ge 0$, $m \ge 2$. Then for each $r \ge m - 1$ and for each graph $G \in \mathcal{G}$ with $p$ vertices there exists partition of $G$ into subgraphs such that
	\begin{itemize}
		\item{The number of vertices in each subgraph does not exceed $r$.}
		\item{The total number of edges mutually connecting subgraphs does not exceed $\frac{\delta p r^{\lambda}}{r}$, where $\delta$ is a constant common for all graphs of the class and for all values of $r$.}
	\end{itemize}
\end{lemma}

We call the corresponding partition of the graph as \emph{$r^{\lambda}$-partition}.

\begin{proof}
The proof of the lemma is similar to the proof of lemma 1 from \cite{RDivision1}. We provide the detailed version of the proof for completeness.

Let $r \ge m - 1$, $G \in \mathcal{G}$, $|V(G)| = p$. If $p \le r$, then the trivial partition containing a single graph $G$ suffices.

Let $p > r$. By the definition of edge $\lambda$-separability graph $G$ can be partitioned into two subgraphs $A$ and $B$ with no more than $\alpha p$ vertices each and no more than $\beta p^{\lambda}$ mutually connecting edges. Since $\mathcal{G}$ is a monotone class, $A, B \in \mathcal{G}$. Thus both $A$ and $B$ can be similarly partitioned into two subgraphs. Let's recursively partition all the subgraphs until we have only pieces with no more than $r$ vertices.

Let's prove that the obtained partition is a $r^{\lambda}$-partition.

The constraint on the number of vertices in subgraphs (no more than $r$ vertices per subgraph) is satisfied by the algorithm of partitioning.

Let $X$ be the total number of edges deleted during the algorithm. We prove an upper bound for $X$. Let's split all subgraphs partitioned at any step of the algorithm into sets $\mathcal{G}_i$ depending on the size of a subgraph. We include into $\mathcal{G}_1$ subgraphs with a size from a half-open interval $(r, r \alpha^{-1}]$. Similarly we include into $\mathcal{G}_2$ subgraphs with a size from a half-open interval $(r \alpha^{-1}, r \alpha^{-2}]$, and so on. If $t = \lceil \log_{\alpha} \frac{r}{p} \rceil$, the last set $\mathcal{G}_t$ includes subgraphs with a size from $(r \alpha^{-(t-1)}, r \alpha^{-t}]$.

Let $1 \le i \le t$. Consider the set $\mathcal{G}_i$. Note that vertex sets of distinct subgraphs from $\mathcal{G}_i$ do not intersect, since the ratio of sizes of such subgraphs is less than $\alpha$. Therefore the total size of all subgraphs in $\mathcal{G}_i$ does not exceed $p$. Hence $|\mathcal{G}_i| \le \frac{p}{r} \cdot \alpha^{i - 1}$. At the same time the total number of edges deleted when partitioning a graph from $\mathcal{G}_i$ does not exceed $\beta (r / \alpha^i)^{\lambda}$.

By summing over all subgraphs from sets $\mathcal{G}_i$ we obtain
\[
    X \le \beta \sum_{i=1}^t \frac{\alpha^{i - 1} p}{r} \left(\frac{r}{\alpha^i}\right)^{\lambda} \le \frac{\beta}{\alpha^{\lambda} (1 - \alpha^{1 - \lambda})} \cdot \frac{p r^{\lambda}}{r}.
\]
This matches the constraint on the number of edges mutually connecting subgraphs of a $r^{\lambda}$-partition.
\end{proof}

We use the following auxiliary notation below. Let $M, S > 0$. We denote by $K(M,S)$ the number of tuples $(x_1, \dots, x_t)$ satisfying the condition
\begin{equation}\label{K_M_S}
1 \le x_i\le M,\qquad \sum_{i=1}^t x_i\le S.
\end{equation}
If $\mathcal{G}$ is a monotone $\lambda$-separable class of graphs with vertex degree bounded by $q$, the properties of $r^{\lambda}$-partition may be stated as follows. Let $\bar{p} = \{p_i\}_{i=1}^t$ be a tuple of sizes of $r^{\lambda}$-partition subgraphs, and let $\bar{s} = \{s_i\}_{i=1}^t$ be a tuple of numbers of edges connecting $r^{\lambda}$-partition subgraphs with the rest of the graph. Then 
\begin{equation}\label{K_partition}
\bar{p} \in K (r, p),\qquad \bar{s} \in K \left(qr, \frac{\delta p r^{\lambda}}{r}\right).
\end{equation}

\section{Lower bound for $\lambda$-separable supports}\label{SecPolSepGraph}

In this section we prove the key result of this paper namely a theorem on the lower bound for all supports from classes $\mathcal{G}(\lambda, q, \theta)$.
\[
	L(M_k^T, n) \gtrsim \max \left(\frac{2^n}{n}, \frac{2^n (1 - \lambda)}{\log k} \right).
\]
Note that the lower bound depends only on the separability function. Parameters $q$ and $\theta$ do not affect the lower bound.

Substantively the proof is obtained in the following way. We partition a subgraph of a support into small fragments. Then we bound the number of Boolean functions computable in the subgraph by the number of Boolean functions computable in the fragments and the number of ways to conduct wires between the fragments.

Since the proof is technically involved we prove several auxiliary lemmas in a separate section \ref{SecPolSepGraphAux}. The proof of the main theorem is finished in section \ref{SecPolSepGraphMain}.

\subsection{Auxiliary lemmas}\label{SecPolSepGraphAux}

The following lemma is an immediate corollary of a classic lemma \cite{Chashkin1}.

\begin{lemma}[\cite{Chashkin1}, p. 198--200]\label{LemL2}
    Let $N(n, m, L)$ be the number of Boolean functions with no more than $n$ inputs, no more than $m$ outputs and the complexity not greater than $L$. Then there exists a constant $c$ such that 
\[
	N(n, m, L) \le \bigl(c(n + L) \bigr)^{n + m + L}.
\]
\end{lemma}

We use $r^{\lambda}$-partitioning of support subgraphs to obtain the lower bound. Technical lemma \ref{LemL5} bounds the number of Boolean operators computable in fragments of a $r^{\lambda}$-partition.

Let $p$ and $s$ be positive integers. Denote by $Z(p, s)$ the number of Boolean functions with no more than $s$ inputs and outputs in total and the complexity not greater than $p$.

Recall the notation $K(M,S)$ introduced in section \ref{SecSep} for sets of tuples satisfying conditions \eqref{K_M_S}. When proving lemma \ref{LemL5} we use the following simple property:
\begin{equation}\label{EqKProp}
    \mbox{If}\quad \bar{x}=\{x_i\}_{i=1}^t \in K(M,S), \quad\mbox{then}\quad \sum_{i=1}^t x_i\log x_i\le S\log M.
\end{equation}

We also use the following inequality for non-negative $x$ and $y$:
\begin{equation}
   (x+y)\log(x+y) \le x\log x + y\log y + x + y, \label{EqLogSum}
\end{equation}
which under the assumption $0 \log 0 = 0$ is a corollary of the binary entropy bound $-a\log a-(1-a)\log(1-a)\le 1$ for $a=\frac{x}{x+y}$.

\begin{lemma}\label{LemL5}
    Let $q \ge 1$, $b \ge 0$, $d > 0$ be constants and let $k \to \infty$ be a parameter. Denote $r = (k \log k)^d$. Let $L$, $M$ be numbers and $\bar{p} = \{p_i\}_{i=1}^t$, $\bar{s} = \{s_i\}_{i=1}^t$, $\bar{u}=\{u_i\}_{i=1}^t$ be tuples satisfying the following conditions:
\begin{gather}
	\bar{p}\in K(r,L),\qquad \bar{s}\in K\left(qkr,\frac{bL}{\log k}\right),\qquad u_i \ge 0,\qquad \sum_{i=1}^t u_i \le M.\label{CondKLem5}
\end{gather}
Then
\begin{equation}\label{DefLemL5}
	\sum_{i=1}^t \log Z(p_i, s_i+u_i) \le d\Biggl(1+O\biggl(\frac{\log\log k}{\log k}\biggr)\Biggr)L \log k+M\log M+O(M).
\end{equation}
\end{lemma} 
\begin{proof}

Using lemma \ref{LemL2}, we have
\[
	Z(p_i, s_i+u_i) \le \bigl(c (p_i + s_i + u_i) \bigr)^{p_i + s_i + u_i}.
\]
Taking the logarithm and summing by all tuple elements, we obtain
\begin{equation}\label{Eq0LemL5}
	\sum_{i=1}^t \log  Z(p_i, s_i+u_i) \le \sum_{i=1}^t (p_i + s_i + u_i) (\log c + \log (p_i + s_i + u_i)).
\end{equation}
Using twice \eqref{EqLogSum}, then \eqref{EqKProp} with \eqref{CondKLem5}, we bound the right side of \eqref{Eq0LemL5}:
\begin{align*}
  \sum_{i=1}^t &(p_i + s_i + u_i) (\log c + \log (p_i + s_i + u_i)) \le\\
  &\le \sum_{i=1}^t(p_i\log p_i+s_i\log s_i + u_i\log u_i + (p_i+s_i+u_i)(\log c+2)) \le\\
  &\le L\log r+\frac{bL}{\log k}\underbrace{\log(qkr)}_{O(\log k)}+M\log M+\underbrace{(\log c+2)\left(L+\frac{bL}{\log k}+M\right)}_{O(L+M)} =\\
  &=L\log r + M\log M + O(L+M).
\end{align*}
Substituting the bound into, \eqref{Eq0LemL5} we obtain
\begin{align*}
    \sum_{i=1}^t \log  Z(p_i, s_i + u_i) &\le L\log r+M\log M+O(L+M)=\\
    &=Ld(\log k+\log \log k) + M\log M + O(L+M)=\\
    &=\Biggl(1+O\biggl(\frac{\log\log k}{\log k}\biggr)\Biggr)Ld\log k+M\log M+O(M).
\end{align*}
\end{proof}

The following two lemmas allow to obtain a trivial lower bound for Shannon function for circuits with arbitrary support.

\begin{lemma}[\cite{Chashkin1}, theorem 11.5]\label{LemN2}
    For each constant $\epsilon > 0$ the ratio of Boolean functions of $n$ variables satisfying the inequality
\[
	L(f) \ge (1 - \epsilon) \frac{2^n}{n},
\]
    approaches $1$ as $n \to \infty$.
\end{lemma} 

\begin{lemma}\label{LemN3}
	Let $T$ be an arbitrary support, $k \in \mathbb{N}$, $n\to\infty$. Then
\[
	L(M_k^T, n) \gtrsim \frac{2^n}{n}.
\]
\end{lemma}
\begin{proof}
The lemma is an immediate corollary of lemma \ref{LemN2} and the fact that the complexity of a multilayer circuit is not less than the size of the corresponding Boolean circuit.
\end{proof}

\subsection{The lower bound theorem}\label{SecPolSepGraphMain}

In this section we finish the proof of lower asymptotic bound for $L(M_k^T, n)$, where $T \in \mathcal{G}(\lambda, q, \theta)$. We also prove a corollary allowing to obtain a lower bound for supports having separability function of more general type, e.\,g. $\log p$, $\sqrt p \log \log p$, etc.

\begin{lemma}\label{LemLowerBoundCount}
    Let $T \in \mathcal{G}(\lambda, q, \theta)$. Let $N_k^T(n, m, L)$ be the number of Boolean functions in $B_{n,m}$ computable by $k$-layer circuits in $T$ with size not greater than $L$. Then as $k \to \infty$, the following inequality holds:
\begin{align*}
	\log N_k^T(n, m, L) &\le \frac{L \log k}{1 - \lambda}\Biggl(1+O\biggl(\frac{\log\log k}{\log k}\biggr)\Biggr) +\\
	&+ (n + m) \bigl(\log L + \log (n + m)\bigr) + O(n + m).
\end{align*}
\end{lemma}

\begin{proof}
Let's denote $r = (k \log k)^{\frac{1}{1 - \lambda}}$. We consider only great enough values of $k$ to suffice the conditions on $r$ from lemma \ref{LemSepDiv}. Thus every finite subgraph $G$ of the support $T$ has a $r^{\lambda}$-partition which we denote by $P(G)$.

We can build a mapping between $k$-layer circuits of size not greater than $L$ computing a Boolean function in $B_{n,m}$ and tuples of the following objects:
\begin{enumerate}
    \item Subgraph $G$ of the support where we embed the corresponding Boolean circuit.
    \item A tuple $\bar v$ of vertices of $G$ where we embed inputs and outputs of the Boolean circuit.
    \item A set of directed wires between fragments of $P(G)$.
    \item A tuple of Boolean functions computed in fragments of $P(G)$.
\end{enumerate}

It is easy to see that circuits computing different Boolean functions are mapped to different tuples. Thus we can bound the number of Boolean functions by the number of possible tuples. Obviously we can bound the number of elements per each tuple item and find the product of the bounds.

Denote the corresponding upper bounds as $A_1, A_2, A_3, A_4$ (here the correspondence is defined by the order of items above).

As $T \in \mathcal{G}(\lambda, q, \theta)$, we have
$$A_1 \le \theta + \theta^2 + \dots + \theta^L \le \frac{\theta^{L+1}}{\theta - 1}.$$

Let's bound $A_2$. Obviously there exist $L^{n+m}$ tuples of $n+m$ vertices of $G$. Hence $A_2 \le L^{n+m}$. 

The number of edges between the fragments of the $r^{\lambda}$-partition $P(G)$ is bounded by $\frac{\delta Lr^{\lambda}}{r} = \frac{\delta L}{k \log k}$, where $\delta$ is a constant. In each of these edges we can conduct no more than $k$ wires. Thus there can be no more than $\frac{\delta L}{\log k}$ wires mutually connecting the fragments of $P(G)$. For each of these wires there are three options: directed in one way, directed in the opposite way and missed. Therefore
$$A_3 \le 3^{\frac{\delta L}{\log k}}.$$

We introduce the following notation to obtain the bound for $A_4$. Let $G_i$ be the fragments of 
$P(G)$, $t$ be the number of the fragments, $p_i$ be the number of vertices in $i$-th fragment, $s_i$ be the number of wires that can be conducted outside from $G_i$, $u_i$ be the total number of inputs and outputs of the Boolean circuit embedded into $G_i$ (i.\,e. the number of items in the tuple $\bar v$ corresponding to the vertices of $G_i$). It's clear that each Boolean function computable in $G_i$ must have no more than $p_i$ gates and no more than $s_i + u_i$ inputs and outputs in total. Hence the following inequality holds:
$$A_4 \le \prod_{i=1}^t Z(p_i, s_i + u_i).$$

By multiplying the bounds for $A_i$ we obtain
$$N_k^T(n, m, L) \le \frac{\theta^{L+1}}{\theta - 1} \cdot L^{n+m} \cdot 3^{\frac{\delta L}{\log k}} \cdot \prod_{i=1}^t Z(p_i, s_i + u_i).$$

Taking the logarithm and omitting the negative addend be obtain
\begin{align}\label{EqMainEstimationLemLowerBoundCount}
	\log N_k^T(n, m, L) &\le  \underbrace{(L + 1) \log \theta}_{O(L)} + (n + m) \log L + \nonumber\\ &+ \underbrace{\frac{\delta L}{\log k} \log 3}_{o(L)} + \sum_{i=1}^t \log Z(p_i, s_i + u_i).
\end{align}

Let's bound the sum $\sum_{i=1}^t \log Z(p_i, s_i + u_i)$ in the right side of \eqref{EqMainEstimationLemLowerBoundCount}. We claim that we can apply lemma \ref{LemL5}.

We obtain the following conditions for tuples $\bar{p} = \{p_i\}_{i=1}^t$, $\bar{s} = \{s_i\}_{i=1}^t$ by using \eqref{K_partition}: $\bar{p} \in K(r, L)$, $\bar{s} \in K\left(qkr, \frac{\delta L}{\log k}\right)$. For tuple $\bar{u} = \{u_i\}_{i=1}^t$ it's obvious that $u_i \ge 0$, $\sum_{i=1}^t u_i \le n + m$. Finally we have $q \ge 1$, $\delta \ge 0$, $\frac{1}{1 - \lambda} > 0$, $k \to \infty$ and $r = (k \log k)^{\frac{1}{1 - \lambda}}$. Thus all conditions of lemma \ref{LemL5} are satisfied. Hence
\begin{align*}
    \sum_{i=1}^t \log Z(p_i, s_i + u_i) &\le \frac{L \log k}{1 - \lambda}\Biggl(1+O\biggl(\frac{\log\log k}{\log k}\biggr)\Biggr) +\\
    &+(n + m) \log (n + m) + O(n + m).
\end{align*}
Combining the bounds for addends in the right side of \eqref{EqMainEstimationLemLowerBoundCount}, we obtain
\begin{align*}
    \log N_k^T(n, m, L) &\le \frac{L \log k}{1 - \lambda}\Biggl(1+O\biggl(\frac{\log\log k}{\log k}\biggr)\Biggr) +\\
    &+ (n + m) \bigl(\log L + \log (n + m)\bigr) + O(n + m).
\end{align*}
\end{proof}

In this paper we are interested in the case when $n + m$ is small compared to $L$. In this case we can simplify the inequality in lemma \ref{LemLowerBoundCount}.
\begin{corollary}
\label{SledLowerBoundCount}
    Under the conditions of lemma \ref{LemLowerBoundCount}, if $k \to \infty$ and $n + m \le L / \log L$, then
\begin{equation*}
	\log N_k^T(n, m, L) \le \frac{L \log k}{1 - \lambda}\Biggl(1+O\biggl(\frac{\log\log k}{\log k}\biggr)\Biggr).
\end{equation*}
\end{corollary}

\begin{lemma}\label{LemLowerBoundMain}
    Let $T \in \mathcal{G}(\lambda, q, \theta)$, $k \to \infty$, $n \to \infty$. Then
%\begin{equation}\label{EqResLemLowerBoundMain}
\[
	L(M_k^T, n) \ge \frac{2^n (1 - \lambda)}{ \log k }\Biggl(1 + O\biggl(\frac{\log\log k}{\log k}\biggr)\Biggr).
\]
%\end{equation}
\end{lemma} 

\begin{proof}
It follows from lemma \ref{LemN3} that for great enough values of $n$ and any $k$ the inequality below holds:
\begin{equation}\label{EqTrivialBoundLemLowerBoundMain}
	L(M_k^T, n) \ge \frac{1}{2} \cdot \frac{2^n}{n}.
\end{equation}

Let $k \to \infty$, $n \to \infty$. We denote $L = L(M_k^T, n)$ for brevity. Using the term $N_k^T(n, m, L)$ defined in lemma \ref{LemLowerBoundCount}, we obtain the identity
$$N_k^T(n, 1, L) = 2^{2^n}.$$

It follows from \eqref{EqTrivialBoundLemLowerBoundMain} that $n = O(\log L) = o(L / \log L)$. By applying corollary \ref{SledLowerBoundCount} of lemma \ref{LemLowerBoundCount} we obtain
\[
    2^n \le \frac{L \log k}{1 - \lambda} \Biggl(1 + O\biggl(\frac{\log\log k}{\log k}\biggr)\Biggr).
\]
This implies to the claim of the lemma.
\end{proof}

\begin{theorem}\label{ThLowerBound}
	Let $0 < \lambda < 1$, $T \in \mathcal{G}(\lambda, q, \theta)$. Then
	\begin{equation*}
	L(M_k^T, n) \gtrsim \max \left(\frac{2^n}{n}, \frac{2^n (1 - \lambda)}{\log k} \right)\quad \mbox{as}\quad k\to\infty,\ n\to\infty.
	\end{equation*}
\end{theorem}
\begin{proof}
    It follows from lemmas \ref{LemN3} and \ref{LemLowerBoundMain}.
\end{proof}

Hereby we obtained the lower bound for $\lambda$-separable supports, i.\,e. for supports with separability function like $p^{\lambda}$. Using theorem \ref{ThLowerBound} one can obtain lower bound for supports with separability function like $\log p$, $\sqrt p \log \log p$, etc.

\begin{corollary}\label{CorollaryLowerBoundLim}
	Let $0 \le \lambda_0 < 1$, $f(p) = O(p^{\lambda})$ for all $\lambda > \lambda_0$. Let $T \in \mathcal{G}(q, \theta)$ be a $f(p)$-separable support. Then
	\begin{equation}\label{EqResCorollaryLowerBoundLim}
	L(M_k^T, n) \gtrsim \max \left(\frac{2^n}{n}, \frac{2^n (1 - \lambda_0)}{\log k} \right)\quad \mbox{as}\quad k\to\infty,\ n\to\infty.
	\end{equation}
\end{corollary}

\begin{proof}
Let $\lambda > \lambda_0$. It's clear that $T$ is $p^{\lambda}$-separable. By theorem \ref{ThLowerBound}, $L(M_k^T, n) \gtrsim \max \left(\frac{2^n}{n}, \frac{2^n (1 - \lambda)}{\log k} \right)$ as $k \to \infty$, $n \to \infty$.

Denote $g(\lambda,k,n)=\max \left(\frac{2^n}{n}, \frac{2^n (1 - \lambda)}{\log k} \right)$. We have
$$\liminf_{\substack{k\to\infty \\ n\to\infty}} \frac{L(M^T_k,n)}{g(\lambda,k,n)}\ge 1\quad\mbox{ as } \quad\lambda_0<\lambda<1.$$
It's easy to see that $\frac{g(\lambda,k,n)}{g(\lambda_0,k,n)}\ge \frac{1-\lambda}{1-\lambda_0}$, therefore
$$A:=\liminf_{\substack{k\to\infty \\ n\to\infty}} \frac{L(M^T_k,n)}{g(\lambda_0,k,n)} \ge \liminf_{\substack{k\to\infty \\ n\to\infty}} \frac{1-\lambda}{1-\lambda_0}\cdot\frac{L(M^T_k,n)}{g(\lambda,k,n)}\ge \frac{1-\lambda}{1-\lambda_0}$$
for all $\lambda_0<\lambda<1$.

Hence $A\ge \sup\limits_{1>\lambda>\lambda_0}\frac{1-\lambda}{1-\lambda_0}=1$. The latter implies to \eqref{EqResCorollaryLowerBoundLim}.
\end{proof}

\section{Lower bound for $d$-dimensional circuits}\label{SecDDimensional}

In this section we prove lower bound for $d$-dimensional circuits and asymptotics for $d$-dimensional rectangular circuits. In substance the lower bound for $d$-dimensional circuits is a corollary for the lower bound for $\lambda$-separable supports, since we prove that all $d$-dimensional supports belong to classes $\mathcal{G}(\lambda, q, \theta)$.

\subsection{Properties of $d$-dimensional graphs}

Indeed we have to prove that $d$-dimensional graphs have three properties: bounded vertex degree, exponentially bounded number of non-isomorphic subgraphs, and $\lambda$-separability.

\begin{lemma}\label{LemDegLimited}
Let $T$ be a $d$-dimensional support (accordingly let $\mathcal{G}$ be a class of $d$-dimensional graphs) with a parameter $c_e$. Then vertex degree of $T$ (accordingly of any graph is $\mathcal{G}$) is bounded by $(2c_e +1)^d$.
\end{lemma}
\begin{proof}
When placing arbitrary $d$-dimensional graph into $d$-dimensional space the neighborhood of any vertex is placed into a ball of radii $c_e$. Since $d$-dimensional balls with radii $0.5$ and centers in graph vertices do not intersect and lie inside a ball with radii $c_e + 0.5$, the number of such balls cannot exceed the ratio of volumes of $d$-dimensional balls with radii $c_e + 0.5$ and $0.5$ respectively. This ratio is equal to $(2c_e +1)^d$.
\end{proof}

\begin{lemma}\label{LemSubgraphExpLimited}
    Let $T$ be a $d$-dimensional graph. Then the number of non-labeled subgraphs of $T$ with $n$ vertices does not exceed $\theta^n$, where $\theta$ is a constant.
\end{lemma}
\begin{proof}
    Immediately follows from the remark to lemma 2 in \cite{Korshunov1}.
\end{proof}

We apply the results of \cite{MultidimensionalSeparator} to prove $\lambda$-separability of $d$-dimensional graphs.

\begin{definition}[\cite{MultidimensionalSeparator}, definition 2.3]
    Let $\alpha \ge 1$ be given, and let $B = \{B_1, B_2, \dots, B_p\}$ be a set of closed $d$-dimensional balls with non-overlapping interiors. The \emph{$\alpha$-overlap graph} for $B$ is the undirected graph with vertices  $V = \{1, 2, \dots, p\}$ and edges 
\[
	E = \{\{i, j\} \colon B_i \cap (\alpha \cdot B_j) \ne \varnothing \text{ and } (\alpha \cdot B_i) \cap B_j \ne \varnothing \},
\]
where $\alpha \cdot B_j$ is a ball centered as $B_j$ and having $\alpha$ times greater radii.
\end{definition}

The following lemma shows the connection between $d$-dimensional graphs and $\alpha$-overlap graphs.
\begin{lemma}\label{LemDDimAlpha}
    Let $\mathcal{G}$ be a class of $d$-dimensional graphs with a parameter $c_e$. Then each graph in $\mathcal{G}$ can be supplemented by some number of edges (maybe $0$) resulting to a $2c_e$-overlap graph in a $d$-dimensional space.
\end{lemma}
\begin{proof}
Let $G \in \mathcal{G}$. Consider the placement of $G$ into a $d$-dimensional space, and let $G'$ be the $2c_e$-overlap graph for the balls of radii $0.5$ and centers in the vertices of $G$. Since the distance between centers of any two balls is not less than $1$, interiors of the balls do not intersect.

If there is an edge between two vertices in $G$, the distance between the centers of the corresponding balls does not exceed $c_e$. Therefore balls with the same centers and radii $0.5$ and $0.5 \cdot 2c_e= c_e$ respectively would intersect. Hence all edges of $G$ are also edges of $G'$.
\end{proof}

\begin{lemma}[\cite{MultidimensionalSeparator}, theorem 2.4]\label{LemDDimensionalSep}
Let $d \ge 1$, $\alpha \ge 1$ be constants. Then there exists a function
\[
	f(p) = O\left(\alpha \cdot p^{\frac{d-1}{d}} + c(\alpha, d)\right)
\]
such that each $\alpha$-overlap graph in a $d$-dimensional space is vertex $f(p)$-separable. The separator splits its parent graph into pieces with no more than $\frac{d+1}{d+2}$ of the initial number of vertices.
\end{lemma}

Essentially lemma \ref{LemDDimensionalSep} claims $\frac{d-1}{d}$-separability of all $\alpha$-overlap graphs in a $d$-dimensional space.

\begin{remark}
In the source \cite{MultidimensionalSeparator} lemma \ref{LemDDimensionalSep} was stated in a slightly different way. Considering any $\alpha$-overlap graph in $d$-dimensional space, it was claimed that the graph has a separator of size bounded by $O\left(\alpha \cdot p^{\frac{d-1}{d}} + c(\alpha, d)\right)$. Since the separability function is common for all graphs in a monotone class, we modified the statement of the lemma in this paper to emphasize the independence of the separability function from individual graphs.
\end{remark}

\begin{corollary}\label{Sl1LemDDimensionalSep}
    Let $\mathcal{G}$ be a class of $d$-dimensional graphs. Then $\mathcal{G}$ is $p^{\frac{d-1}{d}}$-separable.
\end{corollary}
\begin{proof}
Immediately follows from lemmas \ref{LemDDimAlpha} and \ref{LemDDimensionalSep}.
\end{proof}

\subsection{Shannon function bounds}

\paragraph{$d$-dimensional circuits.}
We apply the properties of $d$-dimensional graphs proved in the previous section and obtain the lower bound for $d$-dimensional circuits.
\begin{theorem}\label{ThLowerBoundDDimensional}
    Let $T$ be a $d$-dimensional support. Then
	\begin{equation*}
	L(M_k^T, n) \gtrsim \frac{2^n}{\min(n, d \log k)}\quad \mbox{as}\quad k\to\infty,\ n\to\infty.
	\end{equation*}
\end{theorem}
\begin{proof}
Immediately follows from theorem  \ref{ThLowerBound}, lemmas \ref{LemDegLimited}, \ref{LemSubgraphExpLimited} and corollary \ref{Sl1LemDDimensionalSep}.
\end{proof}

\paragraph{Multidimensional rectangular circuits.}
Multidimensional rectangular circuits are a special case of $d$-dimensional circuits, thus the lower bound from theorem \ref{ThLowerBoundDDimensional} is also applicable for these circuits.

An upper bound of Shannon function for multidimensional rectangular circuits was proved in \cite{KModel1}.

\begin{lemma}[\cite{KModel1}, theorem 1]\label{LemUpperBoundDDimensional}
\[
	L(M_k^d, n) \lesssim \frac{2^n}{\min(n,d \log k)}\quad\mbox{as}\quad k\to\infty,\ n\to\infty.
\]
\end{lemma}

Applying theorem \ref{ThLowerBoundDDimensional} and lemma \ref{LemUpperBoundDDimensional} we obtain the asymptotics of Shannon function for multidimensional rectangular circuits.

\begin{corollary}\label{Sl1AsympDDimensional}
\[
	L(M_k^d, n) \sim \frac{2^n}{\min(n,d \log k)}\quad\mbox{as}\quad k\to\infty,\ n\to\infty.
\]
\end{corollary}

\section{Conclusion}
In this paper we proved the lower bound for Shannon function $L(M_k^T, n) \gtrsim \max \left(\frac{2^n}{n}, \frac{2^n (1 - \lambda)}{\log k} \right)$ for any support $T$ from a class $\mathcal{G}(\lambda, q, \theta)$. An important special case of such supports are $d$-dimensional graphs for which thereby we proved the lower bound $L(M_k^T, n) \gtrsim \frac{2^n}{\min(n, d \log k)}$.

A natural direction of developing the obtained results is examining classes of graphs with a separability function different from $p^{\lambda}$. For example, graphs supporting a placement in a hyperbolic space are of interest. It was proved in \cite{HyperbolicSeparatorTheorem} that such graphs have logarithmic separability function. Corollary \ref{CorollaryLowerBoundLim} of theorem \ref{ThLowerBound} allows to obtain a lower bound for Shannon function for such graphs. However the question about upper bounds remains open.

\end{document}